\documentclass[a4paper,11pt]{article}
\usepackage[margin=1in]{geometry}
\pdfoutput=1 
\usepackage[utf8]{inputenc}
\usepackage[T1]{fontenc}
\usepackage[english]{babel}
\usepackage{csquotes}
\usepackage{amsmath}
\usepackage{amsthm}
\usepackage{amsfonts}
\usepackage{amssymb}
\usepackage{commath}
\usepackage{bm}
\usepackage{xfrac}
\usepackage{mathtools}
\usepackage{authblk}
\usepackage{url}
\usepackage[pdfencoding=auto,psdextra,hidelinks,colorlinks=true, linkcolor=blue, citecolor=blue, urlcolor=blue, pagebackref=true]{hyperref}
\usepackage{bookmark}
\usepackage{orcidlink}
\usepackage{float}
\usepackage{booktabs}
\usepackage{subcaption}
\usepackage{tikz}
\usetikzlibrary{shapes.geometric}
\usepackage{tikz-qtree}
\usepackage{xcolor}
\definecolor{DarkBlue}{RGB}{0,0,139}
\definecolor{RoyalBlue}{RGB}{65,105,225}
\definecolor{DarkRed}{RGB}{139,0,0}
\definecolor{Orange}{RGB}{255,165,0}
\definecolor{Teal}{RGB}{0,128,128}
\usepackage{caption}
\usepackage{nicematrix}
\usepackage{extarrows}
\usepackage[inline]{enumitem}
\usepackage{multicol}

\usepackage{float}
\usepackage{wrapfig}
\usepackage{booktabs}
\usepackage{subcaption}
\usepackage{authblk}
\newcommand*{\email}[1]{\href{mailto:#1}{\nolinkurl{#1}}}

\DeclareMathOperator{\sgn}{sgn}
\DeclareMathOperator*{\argmax}{arg\,max}

\newcommand{\calH}{\ensuremath{\mathcal{H}}}

\newcommand{\Exp}{\operatorname*{E}}

\newcommand{\bfx}{\ensuremath{\mathbf{x}}}

\newcommand{\Simplex}{\ensuremath{\Delta}}

\newcommand{\eps}{\ensuremath{\varepsilon}}
\newcommand{\NN}{\ensuremath{\mathbb{N}}}

\newcommand{\QQ}{\ensuremath{\mathbb{Q}}}
\newcommand{\RR}{\ensuremath{\mathbb{R}}}

\newcommand{\PPAD}{\ensuremath{\mathrm{PPAD}}}

\newcommand{\FIXP}{\ensuremath{\mathrm{FIXP}}}
\newcommand{\UEOPL}{\ensuremath{\mathrm{UEOPL}}}
\newcommand{\LinearFIXP}{\ensuremath{\mathrm{Linear\textrm{-}FIXP}}}

\newcommand{\NOT}{\ensuremath{\mathsf{NOT}}}
\newcommand{\PURIFY}{\ensuremath{\mathsf{PURIFY}}}
\newcommand{\OR}{\ensuremath{\mathsf{OR}}}

\newcommand{\SqrtSum}{\textup{\textsc{SqrtSum}}}

\newcommand{\PureCircuit}{\textup{\textsc{Pure-Circuit}}}

\newcommand{\GCircuit}{\textup{\textsc{GCircuit}}}

\newtheorem{theorem}{Theorem}
\newtheorem{proposition}{Proposition}
\newtheorem{lemma}{Lemma}

\newtheorem{definition}{Definition}

\newcommand{\discount}{\ensuremath{\gamma}}

\newcommand{\auxnot}[1]{\ensuremath{a_\neg^{#1}}}
\newcommand{\auxor}[1]{\ensuremath{a_\lor^{#1}}}
\newcommand{\auxpur}[2]{\ensuremath{a_{\mathrm{P}}^{#1,#2}}}

\begin{document}

\title{On the Complexity of Stationary Nash Equilibria in\\Discounted Perfect Information Stochastic Games}
\author{Kristoffer Arnsfelt Hansen}
\author{Xinhao Nie}
\affil{Aarhus University, Denmark\\\email{arnsfelt@cs.au.dk} \Authand \email{nie@cs.au.dk}}

\date{October 13, 2025}

\maketitle

\begin{abstract}
  We study the problem of computing stationary Nash equilibria in
  discounted perfect information stochastic games from the viewpoint
  of computational complexity. For two-player games we prove the
  problem to be in $\PPAD$, which together with a previous
  $\PPAD$-hardness result precisely classifies the problem as
  $\PPAD$-complete. In addition to this we give an improved and
  simpler $\PPAD$-hardness proof for computing a stationary
  $\eps$-Nash equilibrium. For 3-player games we construct games showing
  that rational-valued stationary Nash equilibria are not guaranteed
  to exist, and we use these to prove $\SqrtSum$-hardness of computing
  a stationary Nash equilibrium in 4-player games.
\end{abstract}

\section{Introduction}
Stochastic games, first introduced in the seminal work of
Shapley~\cite{PNAS:Shapley1953}, are a general model of dynamic
interactions between players. Shapley's initial model is a
discrete-time finite two-player zero-sum game, where in each round of
the game each player independently chooses an action, which results in
the players receiving immediate payoffs (rewards) and a probabilistic
change of state. The overall payoff of a player is determined from the
sequence of rewards by \emph{discounting} according to a discount
factor $\discount<1$. Shapley proved that in such games, the players
have optimal stationary strategies, i.e., they each have a strategy
that to each state describes a probability distribution over the set
of actions of that state, from which the player draws an action each
time the game enters the state. Shapley also considered the case of
\emph{perfect information}, where in each state only one of the
players has more than one action, and noted that in such games the
players have optimal \emph{pure} stationary strategies, i.e.,
strategies where in every state the players always select the same
action.

Shapley's model has since been extended and modified in many ways, and
the resulting models of stochastic games have been studied
extensively~\cite{book:Filar-Vrieze-1996,book:NeymanSorin-2003}. Our
focus will be on the immediate extension to multi-player discounted
stochastic games. Fink~\cite{JSHUA:Fink1964} and
Takahashi~\cite{JSHUA:Takahashi1964} proved the existence of a
stationary Nash equilibrium in such games. Unlike the case of
zero-sum games, pure stationary strategies are not sufficient to
guarantee existence of Nash equilibria in perfect information
games. Indeed, Zinkevich, Greenwald, and
Littman~\cite{NIPS:ZinkevichGL2005-cyclic-equilibria} gave a simple
example of a two-player perfect information nonzero-sum stochastic
game, where each player controls a single state in which they have
just two actions, having a unique \emph{mixed} stationary Nash
equilibrium.

Strategic-form games may be viewed as discounted stochastic games
having a single state that is repeated in every round of play. Optimal
strategies, for the case of zero-sum games, and Nash equilibria, for
the case of nonzero-sum games, of the strategic form game then
correspond to stationary optimal strategies and stationary Nash
equilibria of the single-state stochastic game. From the perspective
of computational complexity, this means that computing stationary
optimal strategies and stationary Nash equilibria is at least as hard
as computing optimal strategies and Nash equilibria in strategic form
games.

The complexity of computing optimal strategies and Nash equilibria in
strategic form games is a well-studied problem. For zero-sum games,
optimal strategies may be computed efficiently using linear
programming. For nonzero-sum games, the works of Daskalakis, Goldberg,
and Papadimitriou~\cite{SICOMP:DaskalakisGP2009-Nash} and Chen, Deng,
and Teng~\cite{JACM:ChenDT2009-Nash} show that computing a Nash
equilibrium in two-player games or computing an $\eps$-Nash
equilibrium in multi-player games is $\PPAD$-complete for polynomially
small $\eps$, and Etessami and
Yannakakis~\cite{SICOMP:EtessamiY2010-FIXP} proved that computing a
Nash equilibrium in multi-player games with at least three players is
$\FIXP$-complete.

While the results for strategic form games provide computational
hardness for computing stationary Nash equilibria in discounted
stochastic games, only recently has the computational complexity been
settled. Deng~et~al.~\cite{NSR:DengLMWY2022-complexity-Markov-perfect}
and Jin, Muthukumar and
Sidford~\cite{ITCS:JinMS2023-complexity-stochastic-games} proved that
computing stationary $\eps$-Nash equilibria is in $\PPAD$, and
Filos-Ratsikas et~al.~\cite{SICOMP:Filos-RatsikasH2023-FIXP} proved
that computing stationary Nash equilibria is in $\FIXP$. The precise
complexity of computing stationary optimal strategies in two-player
zero-sum stochastic games remains open. Etessami and
Yannakakis~\cite{SICOMP:EtessamiY2010-FIXP} proved that the problem is
in $\FIXP$ and is $\SqrtSum$-hard. For the related problem of
approximating the \emph{values},
Batziou~et~al.~\cite{STOC:BatziouFGMS2025-monotone-contractions}
proved that the problem is in the complexity class $\UEOPL$.

For perfect information games, as shown by Andersson and
Miltersen~\cite{ISAAC:AnderssonM2009-games-on-graphs}, the task of
computing optimal strategies in two-player zero-sum games is
polynomial time equivalent to computing optimal strategies in the
model of simple stochastic games introduced by
Condon~\cite{IC:Condon1992-ssg}. This latter problem has been shown to
be contained in \UEOPL~\cite{JCSS:FearnleyGMS2020-UEOPL}, but its
precise complexity remains an elusive open problem.  Recently it was
shown by Jin, Muthukumar and
Sidford~\cite{ITCS:JinMS2023-complexity-stochastic-games} and by
Daskalakis, Golowich, and
Zhang~\cite{COLT:DaskalakisGZ2023-complexity-stochastic-games}, that
computing stationary $\eps$-Nash equilibria in perfect-information
nonzero-sum $\frac{1}{2}$-discounted stochastic games is \PPAD-hard
for some small unspecified constant~$\eps>0$. The proof by Jin,
Muthukumar and Sidford shows $\PPAD$-hardness for games with
polynomially many players, whereas the proof by Daskalakis, Golowich,
and Zhang shows $\PPAD$-hardness even for games with two players.

\subsection{Our Results}
We show the following results for computing stationary Nash equilibria
in perfect-information discounted stochastic games.
\begin{enumerate}
\item For two-player games we show (Theorem~\ref{thm:PPAD-membership})
  that computing stationary Nash equilibria is in $\PPAD$. Taken
  together with the $\PPAD$-hardness result of Daskalakis, Golowich,
  and Zhang~\cite{COLT:DaskalakisGZ2023-complexity-stochastic-games}
  our result thereby establishes that the problem is
  $\PPAD$-complete. As a direct consequence of our result it follows
  that any two-player game has a stationary \emph{rational-valued}
  Nash equilibrium, whenever all numbers defining the stochastic game
  are rational numbers. Such a result may be viewed as a prerequisite
  for solving the problems using pivoting algorithms such as Lemke's
  algorithm~\cite{MS:Lemke1965-Lemke-Algorithm}. Our proof of
  $\PPAD$-membership shows that stationary Nash equilibria can in
  principle be computed by Lemke's algorithm. Previous classes of
  discounted stochastic games known to possess rational-valued
  stationary optimal strategies or Nash equilibria include two-player
  \emph{single-controller}
  games~\cite{JOTA:ParthasarathyR1981-single-controller} and zero-sum
  \emph{switching-controller}
  games~\cite{JOTA:Filar1981-switching-controller}, and our result
  contributes another important class of stochastic games to this line
  of research.
\item We improve the $\PPAD$-hardness result of Daskalakis, Golowich,
  and Zhang, by proving $\PPAD$-hardness for a concrete $\epsilon>0$,
  namely any
  $\eps < \frac{3 - 2 \sqrt{2}}{288} \approx 5.967 \times
  10^{-4}$. The $\PPAD$-hardness proofs by Jin, Muthukumar and Sidford
  and by Daskalakis, Golowich, and Zhang are shown by reduction from
  the so-called $\eps$-$\GCircuit$ problem which, prior to the
  introduction of the \PureCircuit-problem by
  Deligkas~et~al.~\cite{DeligkasFHM22-Pure-Circuit}, was a standard
  way of proving $\PPAD$-hardness. The problem $\eps$-$\GCircuit$ was
  shown to be $\PPAD$-hard for an unspecified constant $\eps>0$ by
  Rubinstein~\cite{SICOMP:Rubinstein2018-inapprox-NE}. The reductions
  from $\eps$-$\GCircuit$ build gadgets for every gate of the given
  generalized circuit, and joining these gadgets together directly
  results in a game with polynomially many players, thus giving the
  result of Jin, Muthukumar and Sidford.  This reduction alone is
  already very involved. Now, assuming a structural property of the
  given generalized circuit, namely that every gate can be assumed to
  have fan-out at most~2, Daskalakis, Golowich, and Zhang observe that
  players may be assigned to gadgets in such a way that the resulting
  game has just~5 players. To obtain their result for two-player
  games, they introduce an intricate notion of \emph{valid} colorings
  of the gates of the generalized circuit and show how to transform a
  given generalized circuit instance into one that allows for such a
  coloring.

  In contrast, we give a very simple and direct reduction
  (Theorem~\ref{thm:epsilon}) from the \PureCircuit-problem to
  two-player games. Reducing from the \PureCircuit\/ problem allows
  for much simpler gadgets and we exploit that $\PPAD$-hardness of
  \PureCircuit\ holds even for circuits with a bipartite
  \emph{interaction graph}, and this enables us to combine the gadgets
  in a natural way.
\item We construct 3-player games (Definition~\ref{def:G(a)}) with
  unique stationary Nash equilibria that are irrational-valued,
  thereby precluding $\PPAD$-membership. We then use these games as
  gadgets to show (Theorem~\ref{thm:sqrt}) that computing a stationary
  Nash equilibrium in 4-player games is $\SqrtSum$-hard. This
  indicates that computing stationary Nash equilibria in games with 3
  or more players brings additional challenges.
\end{enumerate}

\section{Preliminaries}
\label{sec:prelims}

\subsection{Stochastic Games}
We give here a general definition of stochastic games and afterwards
consider the specialization to perfect information games. An infinite
horizon $n$-player finite stochastic game $\Gamma$ is given as
follows. The game is played on a finite set of states $S$. In every
state $k$, each player~$i$ has a set of actions $A_i(k)$. Let
$A(k) = A_1(k) \times \dots A_n(k)$ denote the set of action profiles
in state $k$. Let $P = \{(k,a) \colon k \in S, a \in A(k)\}$ denote
the pairs of states and action profiles of that state. The immediate
payoff, or \emph{reward} to player~$i$ is given by a function
$u_i \colon P \to \RR$ and the state transitions are given by a
function $q \colon P \to \Simplex(S)$, where $\Simplex(S)$ denotes the
set of probability distributions on $S$.

A \emph{play} of $\Gamma$ is an infinite sequence $h \in P^\infty$.  A
\emph{finite play} up to round~$t$ is a sequence
$h_t \in P^{t-1} \times S$. Let
$\calH = \mathbin{\dot{\cup}}_{t=1}^\infty \left(P^{t-1} \times
  S\right)$ denote the set of all finite plays. For a finite play
$h \in \calH$ we denote by $S(h)$ the last element of $h$, i.e., the
current state after the play~$h$. A behavioral strategy for player~$i$
is then a function $\sigma_i \colon \calH \to \Simplex(A_i(S(h)))$
mapping a play $h$ to a probability distribution over $A_i(S(h))$. A
stationary strategy is a behavioral strategy that depends only on the
last state of a finite play, and may thus be viewed as a function
$x_i \colon S \to \Simplex(A_i(k))$ that maps a state $k \in S$ to a
probability distribution over $A_i(k)$. Behavioral strategies
$\sigma_i$ for each player~$i$ form a behavioral strategy profile
$\sigma=(\sigma_1,\dots,\sigma_n)$. In the same way, stationary
strategies for each player form a stationary strategy profile $x$. A
behavioral strategy profile $\sigma$ and an initial state $s^1 \in S$
define, by Kolmogorov's extension theorem, a unique probability
distribution $\Pr_{s^1,\sigma}$ on plays $(s^1,a^1,s^2,a^2,\dots)$,
where the conditional probability of $a^t=a$ given the play up to
round~$t$, $h_t=(s^1,a^1,\dots,s^t)$, is equal to
$\prod_{i=1}^n \Pr[\sigma_i(h_t)=a_i]$, and the conditional
probability of $s^{t+1}$ given $s^t$ and $a^t$ is equal to
$q(s^t,a^t)$. We denote by $\Exp_{s^1,\sigma}$ the expectation with
respect to $\Pr_{s^1,\sigma}$.\medskip

\noindent For every \emph{discount factor} $0\leq \gamma <1$, the
(normalized) $\gamma$-discounted payoff to player~$i$ of play starting
from state~$s^1$ according to $\sigma$ is defined to be
\begin{equation}
  \label{eq:discounted-value}
  V^\gamma_i(s^1,\sigma) = \Exp_{s^1,\sigma} \left[ (1-\gamma) \sum_{t=1}^\infty \gamma^{t-1}u_i(s^t,a^t)\right] \enspace .
\end{equation}
For $\eps \geq 0$, a behavioral strategy profile $\sigma$ is a
$\gamma$-discounted $\eps$-Nash equilibrium for play starting in
state~$s^1$ if
\begin{equation}
  V^\gamma_i(s^1,\sigma) \geq V^\gamma_i(s^1,(\sigma'_i,\sigma_{-i})) -\eps\enspace ,
\end{equation}
for all players $i \in [n]$, and all behavioral strategies
$\sigma'_i$ for player~$i$. Here, $(\sigma'_i,\sigma_{-i})$ denotes
the strategy profile where player~$i$ uses the strategy $\sigma'_i$
and player~$j$, for $j\neq i$, uses the strategy $\sigma_j$. If
$\sigma$ is a $\gamma$-discounted $\eps$-Nash equilibrium for play
starting in state $s^1$ for all $s^1$, we simply say that $\sigma$ is a
$\gamma$-discounted $\eps$-Nash equilibrium. When $\eps=0$ we
say that $\sigma$ is a $\gamma$-discounted Nash equilibrium.

Fink~\cite{JSHUA:Fink1964} and Takahashi~\cite{JSHUA:Takahashi1964}
proved that any finite discounted stochastic game has a
$\gamma$-discounted equilibrium in stationary strategies for any
discount factor~$\gamma$.

When considering $\eps$-Nash equilibria for concrete values of
$\epsilon$ the range of the rewards becomes important. We shall make
the general assumption that all rewards belong to the interval
$[0,1]$. Note that this implies that the payoffs also belong to the
interval $[0,1]$ due to the normalization factor $(1-\gamma)$ in our
definition of payoffs.

A perfect information stochastic game (also known as a
\emph{turn-based} stochastic game) is given by a partition of the
state $S = S_1 \cup \dots S_n$ such that for any pair of players
$i\neq j$ and any state $k \in S_i$ we have $\abs{A_j(k)}=1$. We say
that player~$i$ \emph{controls} the states in $S_i$. We shall simplify
the notation when considering perfect information stochastic
games. For $k \in S_i$ we let $A(k)$ denote the set of actions of
player~$i$. When the game is in state $k \in S$ and action
$a \in A(k)$ is played we let $r^i_{ka}$ denote the reward of
player~$i$ and let $p^{kl}_{a}$ denote the probability that play
continues in state~$l$. A stationary strategy profile is given as
$(x^1,\dots,x^n)$ where $x^i_k$ is the probability distribution of
player~$i$ over the set of actions $A(k)$. We denote by $x^i_{ka}$ the
probability that player~$i$ chooses action~$a$ in state~$k$.

\section{Nash equilibrium in 2-player games}
Our main result of this section establishes $\PPAD$-membership of the problem of computing a stationary Nash equilibrium in 2-player discounted perfect information games. 
\begin{theorem}
\label{thm:PPAD-membership}
  Computing a stationary Nash equilibrium in 2-player discounted perfect information stochastic games is in \PPAD.
\end{theorem}
Combining this with the matching $\PPAD$-hardness result of
Daskalakis~et~al.~\cite{COLT:DaskalakisGZ2023-complexity-stochastic-games}
(that holds even for computing \eps-Nash equilibria) thus establishes
that the problem is \PPAD-complete.

To obtain our result we make use of the framework for proving
$\PPAD$-membership via convex optimization due to
Filos-Ratsikas~et~al.~\cite{STOC:Filos-RatsikasH2024-PPAD}. This
framework builds on the characterization of \PPAD\ in terms of
computing fixed points of piecewise linear (PL) arithmetic circuits
due to Etessami and Yannakakis~\cite{SICOMP:EtessamiY2010-FIXP}. An
arithmetic circuit is a circuit using gates that can perform
arithmetic operations, maximum or minimum, i.e., a gate in
$\{+,-,*,\div, \max,\min\}$, as well as rational constants. A PL
arithmetic circuit restricts the gates to be in
$\{+,-,\max,\min,\times \zeta\}$, where $\times \zeta$ denotes
multiplication by any rational constant $\zeta$. Restricting the
general arithmetic circuits used to define the class $\FIXP$ to PL
arithmetic circuits yields the class $\LinearFIXP$, defined to be
closed under polynomial time reductions. With this, Etessami and
Yannakakis proved that $\PPAD=\LinearFIXP$. This means that to prove
$\PPAD$ membership of a given total search problem, it suffices to
reduce the problem at hand to that of computing a fixed point of a
given PL arithmetic circuit.

Filos-Ratsikas~et~al.~\cite{STOC:Filos-RatsikasH2024-PPAD} defined a
particular type of PL arithmetic circuits, called \emph{PL
  pseudo-circuits} that turn out to be useful for proving
\PPAD-membership.
\begin{definition}
  A PL pseudo-circuit with $n$ inputs and $m$ outputs is a PL
  arithmetic circuit
  $F \colon \RR^n \times [0,1]^\ell \to \RR^m \times [0,1]^\ell$. The
  output of the circuit on input $x$ is any $y$ that satisfies $F(x,z)=(y,z)$
  for some $z \in [0,1]^\ell$.
\end{definition}
A PL pseudo-circuit thus computes a correspondence (multi-function)
$G \colon \RR^n \rightrightarrows \RR^m$. The idea behind the
definition is that when proving $\PPAD$-membership, using the
characterization $\PPAD=\LinearFIXP$, one may construct PL circuits
for computing such a correspondence that is only required to work
correctly at a fixed point, i.e., when the ``auxiliary'' variables $z$
satisfy a fixed point condition.

Filos-Ratsikas~et~al.\ introduced the so-called
\emph{linear-OPT-gate}, implemented as a PL pseudo-circuit, that can
be used in a similar way to the primitive gates in
$\{+,-,\max,\min,\times \zeta\}$, but allows for computing solutions
to certain convex optimization problems. This may in turn be used to
solve \emph{feasibility programs} with \emph{conditional constraints},
and this is the capability we will make use of. The feasibility
program with conditional constraints shown to be solvable using PL
pseudo-circuits by Filos-Ratsikas~et~al.\ are of the form:
\begin{equation}
\label{eq:feasibility-program}
\begin{aligned}
    h_i(y) > 0 \implies  a_i \cdot x \leq b_i& \quad \text{ for } i=1,\dots,m\\   
    x \in [-R,R]^n&
\end{aligned}
\end{equation}
The feasibility program is \emph{parametrized} by $n,m,k \in \NN$, a
rational matrix $A \in \RR^{m\times n}$ with row vectors $a_i$, for
$i=1,\dots,m$, and PL arithmetic circuits $h_i \colon \RR^k \to \RR$,
for $i=1,\dots,m$. It takes as \emph{input} $b \in \RR^m$,
$y \in \RR^k$, and $R \in \RR$, and outputs a feasible solution
satisfying the constraints whenever a feasible solution exists. The
parameters are a fixed part of the PL pseudo-circuit solving the
feasibility program.

A technical tool in our $\PPAD$-membership proof is the following
construction, which we believe could also be useful in other settings
and applications.
\begin{proposition}
\label{prop:selection-circuit}
For any $n \geq 1$, there is a PL pseudo-circuit computing the
correspondence $F \colon [0,1]^n \times [0,1]^n \to [0,1]$ given by
\[
  F(x,y)=[\underline{z},\overline{z}] \enspace ,
\]
where we, for the given input $x$ and $y$, let
$A=\argmax_{j \in [n]} x_j$ and define
$\underline{z} = \min_{j \in A} y_j$ and
$\overline{z} = \max_{j \in A} y_j$. The circuit may be constructed in
time polynomial in $n$.
\end{proposition}
\begin{proof}
  We construct a PL arithmetic circuit computing $F$ by solving a
  sequence of feasibility programs with conditional constraints using
  the linear-OPT gate. Let $x,y \in [0,1]^n$ be the given input. For
  $i \in [n]$, let $A_i = \argmax_{j \in [i]} x_j$ and define
  $\underline{z}_i = \min_{j \in A_i} y_j$ as well as
  $\overline{z}_i = \max_{j \in A_i} y_j$. By this definition we have
  $\underline{z}_1=\overline{z}_1=y_1$,
  $\underline{z}=\underline{z}_n$ and $\overline{z}=\overline{z}_n$.
    
  First the arithmetic circuit computes the values
  $\overline{x}_i=\max_{j \in [i]} x_j$ for $i \in [n]$. Next, for
  each $i = 1,\dots,n$ the arithmetic circuit will compute a value
  $z_i \in [\underline{z}_i,\overline{z}_i]$, thus making $z_n$ the
  desired output of the circuit. For $i=1$, we may take $z_1=y_1$ and
  it is thus sufficient to show how to compute $z_i$ from $z_{i-1}$
  for $i=2,\dots,n$. This is done simply by solving the feasibility
  program:
    \begin{equation}
    \label{eq:selection-feasibility-program}
        \begin{aligned}   
        x_i < \overline{x}_{i-1} \implies &z_i=z_{i-1}\\
        x_i > \overline{x}_{i-1} \implies &z_i=y_i\\
        &z_i \leq \max(z_{i-1},y_i) \\
        &z_i \geq \min(z_{i-1},y_i) 
        \end{aligned}
    \end{equation}
    The feasibility program takes the four inputs $x_i$,
    $\overline{x}_{i-1}$, $z_{i-1}$, and $y_i$, has the single output
    $z_i$, and clearly fits the general form given in
    Equation~\ref{eq:feasibility-program}. More precisely, we may
    express each of the conditional constraints by a pair of
    conditional constraints, expressing the equalities in the
    subsequent by two inequalities. The unconditional constraints may
    be expressed using the constant function~$1$ in the antecedent,
    and we may simply take $R=1$.

    It is now straightforward to prove by induction that
    $z_i \in [\underline{z}_i,\overline{z}_i]$ for all $i \in
    [n]$. Since we have $z_1=y_1$, this clearly holds for
    $i=1$. Assume now that
    $z_{i-1} \in [\underline{z}_{i-1},\overline{z}_{i-1}]$ for
    $1<i\leq n$. If $x_i < \overline{x}_{i-1}$ we have that
    $i \notin A_i$, which means that
    $\underline{z}_i=\underline{z}_{i-1}$ and
    $\overline{z}_i=\overline{z}_{i-1}$. The unique solution of the
    feasibility program is $z_i=z_{i-1}$ and thus
    $z_i \in [\underline{z}_i,\overline{z}_i]$. If instead
    $x_i > \overline{x}_{i-1}$ we have $A_i=\{i\}$, which means that
    $\underline{z}_i=\overline{z}_i=y_i$ The unique solution of the
    feasibility is now $z_i=y_i$ and thus
    $z_i \in [\underline{z}_i,\overline{z}_i]$. Finally, if
    $x_i=\overline{x}_{i-1}$ we have $A_i = A_{i-1} \cup \{i\}$ which
    means that $\underline{z}_i=\min(\underline{z}_{i-1},y_i)$ and
    $\overline{z}_i=\max(\overline{z}_{i-1},y_i)$. The solutions of
    the feasibility program form the interval
    $[\min(z_{i-1},y_i),\max(z_{i-1},y_i)]$ which is a subinterval of
    $[\underline{z}_i,\overline{z}_i]$.
\end{proof}

Since the feasibility program used in the proof above only has four
inputs and one output and does not use the full capabilities of the
linear-OPT gate, one may also directly construct a relatively simple
PL pseudo-circuit for solving it, and we give such a construction in
Appendix~\ref{sec:direct-PL-pseudo-circuit}.

\subsection{Proof of \PPAD-membership}

Consider a two-player perfect information stochastic game given as
defined in Section~\ref{sec:prelims}.  We shall without loss of
generality assume that $r^i_{ka}>0$ for every $k \in S$, $a \in A(k)$,
and $i \in \{1,2\}$.

A \emph{valuation} is a pair of vectors $(v^1,v^2) \in \RR^S \times \RR^S$. 
Let $(v^1,v^2)$ be a valuation and $(x^1,x^2)$ a stationary strategy
profile. For $i \in \{1,2\}$, every $k \in S$ and $a \in A(k)$ define
\emph{action valuations} $v^i_{ka}$ by
\begin{equation}\label{eq:action-valuations}
  v^i_{ka} = r^i_{ka} + \discount \sum_{l \in S} p^{kl}_{a} v^i_l \enspace .
\end{equation}

Based on the given valuation and stationary strategy profile
we may compute \emph{updated} valuations $(\widetilde{v}^1,\widetilde{v}^2)$
by

\begin{equation}\label{eq:updated-valuations-1}
  \widetilde{v}^1_k = \begin{cases}  \max_{a \in A(k)} v^1_{ka} & \text{if } k \in S_1 \\ \sum_{a \in A(k)} x^{2}_{ka} v^1_{ka} & \text{if } k \in S_2
  \end{cases} \phantom{\enspace .}
\end{equation}
and
\begin{equation}\label{eq:updated-valuations-2}
  \widetilde{v}^2_k = \begin{cases}  \max_{a \in A(k)} v^2_{ka} & \text{if } k \in S_2 \\ \sum_{a \in A(k)} x^{1}_{ka} v^2_{ka} & \text{if } k \in S_1
  \end{cases} \enspace .
\end{equation}

A stationary strategy profile $(x^1,x^2)$ induces a \emph{unique}
valuation $(v^1,v^2)$ that is a fixed point solution of
Equations~\eqref{eq:updated-valuations-1} and
\eqref{eq:updated-valuations-2}. That is, it induces a valuation $(v^1,v^2)$ that
equals its own updated valuation according to $(x^1,x^2)$.

We say that a stationary strategy profile
$(\widetilde{x}^1,\widetilde{x}^2)$ is \emph{one-step optimal} with respect to
$(v^1,v^2)$ if for every $k \in S$ and $i \in \{1,2\}$ we have
\begin{equation}
  \widetilde{x}^i_{ka}>0 \implies v^i_{ka} = \max_{a' \in A(k)} v^i_{ka'} \enspace .
\end{equation}

The equations above give rise to a correspondence $F$ mapping a pair,
consisting of a valuation $(v^1,v^2)$ and stationary strategy profile $(x^1,x^2)$,
to the set of pairs consisting of the updated valuations
$(\widetilde{v}^1,\widetilde{v}^2)$ and one-step optimal strategy
profiles $(\widetilde{x}^1,\widetilde{x}^2)$. The fixed points of $F$
corresponds exactly to stationary Nash equilibrium strategy
profiles~\cite{JSHUA:Takahashi1964}.

It is not possible to compute the correspondence $F$ by a PL
arithmetic circuit due to the products $x^2_{ka} v^1_{ka}$ and
$x^1_{ka} v^2_{ka}$ in Equation~\eqref{eq:updated-valuations-1} and
Equation~\eqref{eq:updated-valuations-2}, which by
Equation~\eqref{eq:action-valuations} would involve products of
variables of the form $x^2_{ka}v^1_l$ and $x^1_{ka}v^2_l$. Now, for
2-player games we can partly circumvent this obstacle, just by noting
that when given a valuation $(v^1,v^2)$, for which we \emph{know} it is
a valuation of a stationary Nash equilibrium, we may efficiently
compute a stationary strategy profile that both induces the valuation
$(v^1,v^2)$ and is one-step optimal with respect to $(v^1,v^2)$.

Namely, suppose we are given such a valuation $(v^1,v^2)$. Compute the
corresponding action valuations $v^i_{ka}$ for all $k \in S$,
$a \in A(k)$ and $i \in \{1,2\}$. By assumption we have
$v^i_k = \max_{a \in A(k)} v^i_{ka}$ for all $k \in S$, $a \in A(k)$,
and $i \in \{1,2\}$. For $k \in S_i$, let now
$B_{(v^1,v^2)}(k) = \argmax_{a \in A(k)} v^i_{ka}$ denote the set of
one-step optimal actions of player~$i$ in state $k$. To find
$(x^1,x^2)$ we may then just solve the following
system of linear inequalities.

\begin{equation}\label{eq:recover-NE-from-valuations}
  \begin{aligned}
    \sum_{a \in A(k)} x^{2}_{ka} v^1_{ka} & = v^1_k && k \in S_2 \\
    \sum_{a \in A(k)} x^{1}_{ka} v^2_{ka} & = v^2_k &&  k \in S_1 \\
    \sum_{a \in A(k)} x^i_{ka} & = 1 && i \in \{1,2\}, k \in S_i\\
    x^i_{ka} & \geq 0 && i \in \{1,2\}, k \in S_i,  a \in B_{(v^1,v^2)}(k)\\
    x^i_{ka} & = 0 && i \in \{1,2\}, k \in S_i,  a \notin B_{(v^1,v^2)}(k)    
  \end{aligned}
\end{equation}

This leaves the problem of finding a valuation $(v^1,v^2)$ that is
induced by a stationary Nash equilibrium strategy profile. To do this,
we define a correspondence $G$ mapping valuations $(v^1,v^2)$ to
valuations $(\widetilde{v}^1,\widetilde{v}^2)$ as follows. First,
given $(v^1,v^2)$, compute all actions valuations $v^i_{ka}$ for
$i \in \{1,2\}$, every $k \in S$ and $a \in A(k)$. Following
Equation~\eqref{eq:updated-valuations-1} and
Equation~\eqref{eq:updated-valuations-2} we may immediately compute
$\widetilde{v}^i_k = \max_{a \in A(k)} v^i_{ka}$ for $i \in \{1,2\}$
and every $k \in S_i$.

For $k \in S_1$ let
\begin{equation}
  \underline{v}^2_k = \min_{a \in B_{(v^1,v^2)}(k)} v^2_{ka} \quad \text{and} \quad
\overline{v}^2_k = \max_{a \in B_{(v^1,v^2)}(k)} v^2_{ka} \enspace ,
\end{equation}
and similarly, for $k \in S_2$ let
\begin{equation}
  \underline{v}^1_k = \min_{a \in B_{(v^1,v^2)}(k)} v^1_{ka} \quad \text{and} \quad
\overline{v}^1_k = \max_{a \in B_{(v^1,v^2)}(k)} v^1_{ka} \enspace .
\end{equation}
The possible function values of $G$ are then given by any $\widetilde{v}^2_k \in [\underline{v}^2_k , \overline{v}^2_k]$ for $k \in S_1$ and any $\widetilde{v}^1_k \in [\underline{v}^1_k , \overline{v}^1_k]$ for $k \in S_2$.
A fixed point $(v^1,v^2)$ of $G$ implies that the system of
inequalities~\eqref{eq:recover-NE-from-valuations} is feasible, and we
may thus from $(v^1,v^2)$ compute a stationary Nash equilibrium
strategy profile in polynomial time using linear programming. We can compute the correspondence
$G$ by a PL pseudo-circuit using the construction of Proposition~\ref{prop:selection-circuit}.

This implies that we can compute the valuation of a stationary Nash
equilibrium strategy profile in \PPAD, and since a corresponding
stationary Nash equilibrium strategy profile can be computed in
polynomial time from the valuation, the problem of finding a
stationary Nash equilibrium strategy profile is also in \PPAD.

\section{Approximate Nash Equilibrium in 2-player games}

In this section, we show \PPAD\/-hardness of the problem of computing
an approximate stationary Nash equilibrium in 2-player discounted
perfect-information games by reducing from the $\PPAD$-complete
problem \PureCircuit, introduced by
Deligkas~et~al.~\cite{DeligkasFHM22-Pure-Circuit}.  It follows, as a
straightforward consequence, that computing an \emph{exact} Nash
equilibrium is \PPAD\/-hard as well.

The following definition is a special version of the definition by
Deligkas~et~al.~\cite{DeligkasFHM22-Pure-Circuit} suited for our
application.
\begin{definition}[\PureCircuit\/ problem \cite{DeligkasFHM22-Pure-Circuit}]
  An instance $I$ of \PureCircuit\/ is given by a vertex set $V = [n]$
  and a set $G$ of gates. Each gate $g \in G$ belongs to one of three
  types $\{\NOT, \OR, \PURIFY \}$. Given such an instance, the task is
  to find an assignment $\bfx \in \{0,1,\bot\}^V$ that satisfies the
  constraints of each gate described below.
  \begin{itemize}
  \item A gate of type $\NOT$ is given as $(\NOT, u, v)$ with
    input node $u$ and output node $v$ and places the following
    constraint on $\bfx$:
    \[
      (\bfx[u] = 0 \implies \bfx[v] = 1) \land (\bfx[u] = 1 \implies \bfx[v] = 0) \enspace .
    \]
  \item A gate of type $\OR$ is given as $(\OR, u, v, w)$
    with input nodes $u,v$ and output node  $w$ and places the following
    constraint on $\bfx$:
    \[
      (\bfx[u] = \bfx[v] = 0 \implies \bfx[w] = 0)\land ((\bfx[u] = 1)
      \lor (\bfx[v] = 1) \implies \bfx[w] = 1) \enspace .
    \]
  \item A gate of type $\PURIFY$ is given as $(\PURIFY, u, v, w)$
    with input node $u$ and output nodes $v, w$ and places the following
    constraint on $\bfx$:
    \[
      (\{\bfx[v],\bfx[w] \} \cap \{0,1\} \neq \emptyset) \land
      (\bfx[u] \in \{0,1\} \implies \bfx[v] = \bfx[w] = \bfx[u])
      \enspace .
    \]
  \end{itemize}
  Each node of $I$ is the output node of at most one gate. The
  \emph{interaction graph} of $I$ is the directed graph with vertex
  set $V$ and having an edge from node $u$ to node $v$ if there is a
  gate having $u$ as input node and $v$ as output node.
  \end{definition}

  Deligkas~et~al.~\cite[Corollary~2.3]{DeligkasFHM22-Pure-Circuit}
  proved that the \PureCircuit\ problem defined above is
  \PPAD-complete, even when assuming that the interaction graph is
  bipartite. This property is crucial for our simple reduction to
  2-player games, as detailed in the proof below.

\begin{theorem}
  \label{thm:epsilon}
  For any $0 \leq \epsilon < \frac{3 - 2 \sqrt{2}}{288}$, the problem
  of computing an $\epsilon$-approximate Nash equilibrium for 2-player
  $\frac{1}{2}$-discounted perfect information stochastic games is
  \PPAD-hard. This holds even when every player has at most~$2$
  actions in every state of the game and where players strictly
  alternate being the controlling player in any play.
\end{theorem}

\begin{proof}
  We construct a reduction from the \PureCircuit\ problem with \NOT\
  gates, \OR\ gates, and \PURIFY\ gates, having a bipartite
  interaction graph. Let $I$ be a given \PureCircuit\ instance with
  nodes $V$, with bipartition, $V = V_1 \dot\cup V_2$ and gates
  $G$. From $I$ we construct a stochastic game $\Gamma_I$ as follows.
  
  The states of $\Gamma_I$ are formed by the set of nodes $V$,
  controlled by the two players according to the bipartition, together
  with a constant number of auxiliary states.  Thus, letting
  $S = S_1 \dot\cup S_2$ denote the partition of the states of the two
  players, we have $V_1 \subset S_1$ and $V_2 \subset S_2$. For a
  state $u \in V_i$, player~$i$ is given the action set $\{0,1\}$ and
  for an auxiliary state $u \in S_i \setminus V_i$, player~$i$ has the
  (trivial) action set $\{1\}$. A stationary strategy may thus be
  described by numbers $p_u \in [0,1]$, where $p_u$ for $u \in V_i$ is
  the probability of player~$i$ choosing action~$1$.

  A stationary strategy given by $(p_u)_{u \in S}$ is mapped to an
  assignment $\bfx$ of $I$ according to parameters $l$ and $r$, such
  that $0<l<r<1$, to be specified later. If $p_u \in [0,l]$ we let
  $\bfx[u]=0$, if $p_u \in [r,1]$ we let $\bfx[u]=1$, and otherwise we
  let $\bfx[u]=\bot$.

  The rewards of the players in $\Gamma_I$ all belong to the set~$[0,1]$
  and we ensure that for $i \in \{1,2\}$ and all states $u \in S_i$,
  player~$i$ is given reward~$0$ for both actions in state~$u$. In
  addition we ensure that action~$0$ gives reward~0 to both players. This property,
  together with having players alternating, allows us to bound future
  discounted rewards to player~$i$ starting from a state controlled by
  player~$i$ beyond the following state.
  \begin{lemma}
    \label{lemma:future-bound}
    Let $s^1 \in S_i$ and let $\sigma$ be any strategy profile. Then
    \begin{equation}
      \frac{1}{2} \Exp_{s^1,\sigma}[u_i(s^2,a^2)] \leq  V_i^{\frac{1}{2}}(s^1,\sigma) \leq \frac{1}{2} \Exp_{s^1,\sigma}[u_i(s^2,a^2)] + \frac{1}{12} \enspace ,
    \end{equation}
    where the play starting at state $s^1$ given by $\sigma$ is denoted as $(s^1,a^1,s^2,a^2,\dots)$.
  \end{lemma}
  \begin{proof}
    Since the players are alternating and player~$i$ can only get
    non-zero reward at states controlled by the other player, the
    terms with odd $t$ in Equation~\ref{eq:discounted-value} are all~0
    and the terms with even $t>2$ are bounded by $(\frac{1}{2})^t$,
    which means that they in total sum up to at most
    $(\frac{1}{2})^4 + (\frac{1}{2})^6 + (\frac{1}{2})^8 \dots \leq \frac{1}{12}$.
  \end{proof}
  
  To simulate an absorbing state giving each player reward~0, while
  having players alternate, we create a simple 2-cycle between an
  auxiliary state for each player in which both players receive
  reward~0. In the following figures we indicate this simply by the
  pair $(0,0)$, with the understanding that an edge pointing to
  $(0,0)$ is actually pointing to the auxiliary state of the cycle to
  maintain alternation between the players. We denote this as the
  absorbing cycle.

  We next construct parts of $\Gamma_I$ to simulate the gates of $I$
  depending on the type of the gates. For each player and every type
  of gate we have an auxiliary state for each output of the gate
  type. We denote the auxiliary states for player $i$ by \auxnot{i}
  and \auxor{i} for types \NOT\ and \OR\, and by \auxpur{i}0 and
  \auxpur{i}1 for type \PURIFY. The rewards in the auxiliary states
  are given in the analysis below.

  We illustrate these for each type of gate in
  Figure~\ref{fig:reduction}, where gray circular nodes and red square
  nodes are used to distinguish the players, and the orange and blue
  arcs are used to distinguish the actions.  For simplicity of
  notation and analysis, we assume, for the gate under consideration,
  that player~1 is controlling the output nodes, which means that
  player~2 is controlling the input nodes. The case of player~2
  controlling the output node is constructed in the exact same way,
  with the roles of player~1 and player~2 exchanged.

  \paragraph{\bf \NOT\ gates.} Consider a gate $g = (\NOT, u, v)$. In
  state $v$, controlled by player~1, action~0 leads to the state $u$
  and action~1 leads to the auxiliary state $\auxnot2$ which in turn
  leads to the absorbing cycle. Let $r_{\auxnot2}$ denote the reward
  to player~1 in state $\auxnot2$. Let $V$ denote the payoff to
  player~1 of play starting in state~$v$, and let $V_0$ and $V_1$
  denote the payoff to player~1 obtained by switching action in
  state~$v$ to~0 and~1, respectively. We then have
  $V=p_v V_1 + (1-p_v)V_0$.

  If player~1 chooses action~0 in state~$v$, player~1 receives reward
  $p_u$ in the state $u$, and by Lemma~\ref{lemma:future-bound} we
  have $V_0 \in [\frac{1}{4}p_u, \frac{1}{4}p_u+\frac{1}{12}]$. If
  player~1 instead chooses action~1, player~1 receives reward
  $r_{\auxnot2}$ in state $\auxnot2$ and since play continues in the
  absorbing cycle we have $V_1= \frac{1}{4} r_{\auxnot2}$.

  To ensure that the \NOT\ gate is satisfied, we need to 
  ensure that \begin{enumerate*}[label=(\roman*)]
  \item
    \label{not:con1}
    $ p_v \geq r$ when $p_u \leq l$, and
  \item
    \label{not:con2}
    $p_v \leq l$ when $p_u \geq r$.
  \end{enumerate*}
  We show in Appendix~\ref{apdx：NOT-reduction} that if
  \begin{equation}
    \label{eq:bound1}
  \left(\frac{1}{1-r} + \frac{1}{l}\right)2\epsilon <
  \frac{1}{2}r - \frac{1}{2}l - \frac{1}{6} \enspace ,
  \end{equation}
  we may assign a rational value to $r_{\auxnot2}$ such that the above
  conditions are satisfied by any stationary $\eps$-approximate Nash
  equilibrium.

  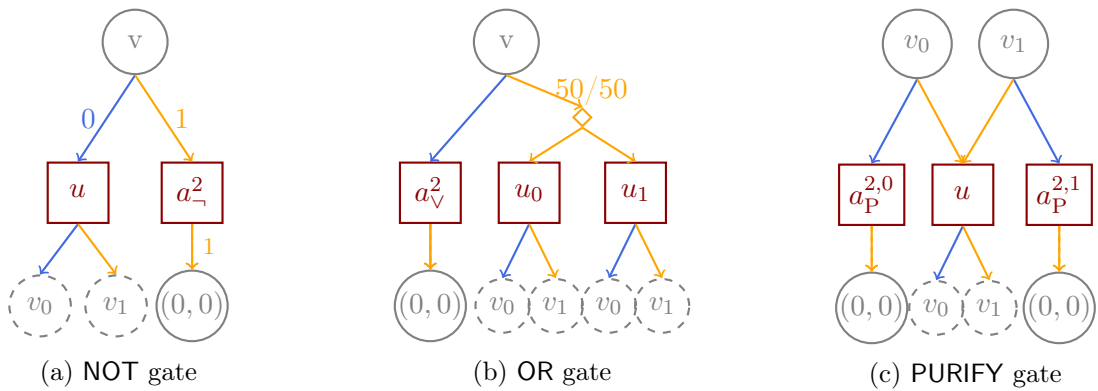
\begin{figure}[ht]
  \centering
  \begin{subfigure}[t]{0.3\textwidth}
    \centering
    \begin{tikzpicture}[thick,->,baseline=(current bounding box.north)]
      \node [draw, circle, color=gray, inner sep =6pt]{v}
    child[color=RoyalBlue] {
    node[draw,color=DarkRed, minimum size=0.8cm, yshift=-5mm] {$u_{}$}
    [sibling distance = 10mm]
    child[color=RoyalBlue] {
      node[draw, dashed, circle, color=gray, inner sep =4pt] {$v_0$}
    }
    child[color=Orange] {
      node[draw, dashed, circle, color=gray, inner sep =4pt] {$v_1$}  
    }
    edge from parent
    node[left] {0}; 
    }
    child[color=Orange] {
    node[draw,color=DarkRed,minimum size=0.8cm, yshift=-5mm] (top) {$\auxnot2$}
    child[color=Orange]{
      node[draw, circle, color=gray, inner sep =0pt](bot) {$(0,0)$}
    }
    edge from parent
    node[right] {1}
    };
  \draw[Orange, dash pattern=on 2pt off 2pt] (top) -- (bot)
    node[right,midway,color=Orange] {\footnotesize 1}
    ;
    \end{tikzpicture}
    \caption{\NOT\ gate}
  \end{subfigure}
  \hfill
  \begin{subfigure}[t]{0.3\textwidth}
    \centering
    \begin{tikzpicture}[thick,->,baseline=(current bounding box.north)]
        \node
        [draw, circle, color=gray, inner sep =6pt] {v}
        [sibling distance = 20mm]
        
        child[color=RoyalBlue] {
          node[draw, color=DarkRed,yshift=-5mm, minimum size=0.8cm](top) {$\auxor2$}
          child[color=Orange]{
          node[draw, circle, color=gray, inner sep =0pt](bot) {$(0,0)$}
            }
        }
        child[color=Orange, yshift=5mm] {
            node[draw, diamond, minimum size = 7pt,inner sep=0pt] {}
            [sibling distance = 14mm]
            child[color=Orange] {
                node[draw, color=DarkRed, minimum size=0.8cm, yshift=5mm] {$u_0$}
                [sibling distance = 7mm]
                child[color=RoyalBlue] {
                    node[draw, dashed, circle, color=gray, inner sep =3pt] {$v_0$}
                }
                child[color=Orange] {
                    node[draw, dashed, circle, color=gray, inner sep =3pt] {$v_1$}  
                }
            }
            child[color=Orange] {
                node[draw, color=DarkRed,minimum size=0.8cm, yshift=5mm] {$u_1$}
                [sibling distance = 7mm]
                child[color=RoyalBlue] {
                node[draw, dashed, circle, color=gray, inner sep =3pt] {$v_0$}
            }
            child[color=Orange] {
            node[draw, dashed, circle, color=gray, inner sep =3pt] {$v_1$}  
            }
            }
            edge from parent
            node[right] {50/50}
        }        ;
        \draw[Orange] (top) -- (bot);
        \draw[Orange, dash pattern=on 2pt off 2pt] (top) -- (bot);
    \end{tikzpicture}
    \caption{\OR\ gate}
  \end{subfigure}
  \hfill
  \begin{subfigure}[t]{0.3\textwidth}
    \centering
    \begin{tikzpicture}[thick,->,baseline=(current bounding box.north)]
        \node[circle,draw, inner sep=5pt, color=gray] {$v_0$}
        [sibling distance=12mm]
        child[color=RoyalBlue] 
        { 
        node[draw,color=DarkRed, yshift=-5mm, minimum size=0.8cm] (x1){$\auxpur{2}0$}
        child[Orange]{
        node[circle,draw,color=gray, inner sep=0pt] (y1) {$(0,0)$}
        }
        }
        child[color=Orange] { 
        node[draw,color=DarkRed,minimum size=0.8cm,yshift=-5mm] (x) {$u_{}$} 
        [sibling distance = 7mm]
            child[color=RoyalBlue] {
            node[draw, dashed, circle, color=gray, inner sep =3pt] {$v_0$}
            }
            child[color=Orange] {
            node[draw, dashed, circle, color=gray, inner sep =3pt] {$v_1$}  
            }
        }
        ;
        
        \node[circle,draw, right=8mm, inner sep=5pt, color=gray](y) {$v_1$}
        [sibling distance=12mm]
        child [missing]
        child [color=RoyalBlue]{ 
        node[draw, color=DarkRed, yshift=-5mm,minimum size=0.8cm] (x2){$\auxpur{2}1$} 
        child[Orange]{
        node[circle,draw,color=gray, inner sep = 0pt] (y2) {$(0,0)$}
        }
        }

        ;
        \draw[Orange] (y.south) -- (x.north);
        \draw[Orange, dash pattern=on 2pt off 2pt] (x1) -- (y1);
        \draw[Orange, dash pattern=on 2pt off 2pt] (x2) -- (y2);
    \end{tikzpicture}
    \caption{\PURIFY\ gate}
  \end{subfigure}

  \caption{Games for the three gates, when player~1 controls the outputs.}
  \label{fig:reduction}
  \end{figure}

  \paragraph{\bf \OR\ gates.} Consider a gate $g = (\OR, u_0,
  u_1,v)$. In state $v$, controlled by player~1, action~0 leads to the
  auxiliary state $\auxor2$ which in turn leads to the absorbing
  cycle. When player~1 chooses action~1, the next state is $u_0$ with
  probability $\frac{1}{2}$ and $u_1$ with probability $\frac{1}{2}$.
  Let $r_{\auxor2}$ denote the reward to player~1 in state
  $\auxor2$. Similarly to the case above, we let $V$ denote the payoff
  to player~1 of play starting in state~$v$, and let $V_0$ and $V_1$
  denote the payoff to player~1 obtained by switching action in
  state~$v$ to~0 and~1, respectively. We again have
  $V=p_v V_1 + (1-p_v)V_0$.

  If player~1 chooses action~0, player~1 receives reward $r_{\auxor2}$
  in state $\auxnot2$ and since play continues in the absorbing cycle
  we have $V_0= \frac{1}{4} r_{\auxor2}$.  If player~1 instead chooses
  action~1 in state~$v$, player~1 receives reward $p_{u_0}$ or
  $p_{u_1}$ in the state $u_0$ or $u_1$, each with
  probability~$\frac{1}{2}$, and by Lemma~\ref{lemma:future-bound} we
  have
  $V_1 \in [\frac{1}{8}(p_{u_0}+p_{u_1}),
  \frac{1}{8}(p_{u_0}+p_{u_1})+\frac{1}{12}]$.

  To ensure that the \OR\ gate is satisfied, we need to 
  ensure that \begin{enumerate*}[label=(\roman*)]
  \item
    \label{or:original-con1} $ p_v \leq l$ when both $p_{u_0} \leq l$ and
    $p_{u_1}\leq l$, and
  \item
    \label{or:original-con2} $p_v \geq r$ when either $p_{u_0} \geq r$ or
    $p_{u_1}\geq r$.
  \end{enumerate*}
  We shall in fact ensure the stronger conditions, that
  \begin{enumerate*}[label=(\roman*')]
  \item
    \label{or:con1} $ p_v \leq l$ when $p_{u_0} + p_{u_1} \leq 2l$,
    and
  \item
    \label{or:con2} $p_v \geq r$ when $p_{u_0} + p_{u_1}\geq r$.
  \end{enumerate*}
  We show in Appendix~\ref{apdx：OR-reduction} that if
  \begin{equation}
    \label{eq:bound2}
    \left(\frac{1}{1-r} + \frac{1}{l}\right)2\epsilon <
    \frac{1}{4}r - \frac{1}{2}l - \frac{1}{6}
  \end{equation}
  we may assign a rational value to $r_{\auxor2}$ such that the above
  conditions are satisfied by any stationary $\eps$-approximate Nash
  equilibrium.
   
  \paragraph{\bf \PURIFY\ gates.} Consider a gate
  $g = (\PURIFY, u, v_0,v_1)$. In state $v_0$, controlled by player~1,
  action~1 leads to the state $u$ and action~0 leads to the auxiliary
  state $\auxpur20$ which in turn leads to the absorbing cycle. In
  state $v_1$, controlled by player~1, action~1 leads to the state $u$
  and action~0 leads to the auxiliary state $\auxpur21$ which in turn
  leads to the absorbing cycle. Let $r_{\auxpur20}$ and
  $r_{\auxpur21}$ denote the rewards to player~1 in state $\auxpur20$
  and $\auxpur21$. The states $v_0$ and $v_1$ behave similarly to the
  state $v$ in the case of a \NOT\ gate, but with different rewards in
  the auxiliary states. Hence the analysis of each state is similar as
  well. For the full analysis, we introduce an additional parameter
  $m$, where $l < m <r$.

  To ensure that the \PURIFY\ gate is satisfied, it is sufficient to
  ensure that
  \begin{multicols}{2}
    \begin{enumerate}[label=(\roman*)]
    \item \label{pur:con1} $p_{v_0} \leq l$ when $p_u \leq m$
    \item \label{pur:con2} $p_{v_0} \geq r$ when $p_u \geq r$
    \item \label{pur:con3} $p_{v_1} \leq l$ when $p_u \leq l$
    \item \label{pur:con4} $p_{v_1} \geq r$ when $p_u \geq m$
    \end{enumerate}
  \end{multicols}
  In case $p_u \leq l$, conditions~\ref{pur:con1} and~\ref{pur:con3}
  give that both $p_{v_0} \leq l$ and $p_{v_1} \leq l$. In case
  $p_u \geq r$, conditions~\ref{pur:con2} and~\ref{pur:con4} gives
  that both $p_{v_0} \geq r$ and $p_{v_1} \geq r$. Finally, for any value
  of $p_u$, at least one of the conditions~\ref{pur:con1}
  or~\ref{pur:con4} is satisfied, which gives either $p_{v_0}\leq l$
  or $p_{v_1}\geq r$.

  We show in Appendix~\ref{apdx：PURIFY-reduction} that if both
  \begin{equation}
    \label{eq:bound3}
    \left(\frac{1}{1-r} + \frac{1}{l}\right)2\epsilon <
    \frac{1}{2}r - \frac{1}{2}m - \frac{1}{6}
  \end{equation}
  and
  \begin{equation}
    \label{eq:bound4}
    \left(\frac{1}{1-r} + \frac{1}{l}\right)2\epsilon <
    \frac{1}{2}m - \frac{1}{2}l - \frac{1}{6}
  \end{equation}
  we may assign rational values to $r_{\auxpur20}$ and $r_{\auxpur21}$
  such that the above conditions are satisfied by any stationary
  $\eps$-approximate Nash equilibrium.

  We show in the lemma below that we may find appropriate
  constants $l$, $r$, and $m$ satisfying all the required conditions
  above. This completes the proof, since this also means that the
  reduction may also be carried out in polynomial time.
    
  \begin{lemma}
    \label{lemma:epsilon}
    There exist constants $l$, $r$, and $m$ such that for any
    $\epsilon < \frac{3 - 2 \sqrt{2}}{288}$ conditions
    \eqref{eq:bound1}, \eqref{eq:bound2}, \eqref{eq:bound3} and
    \eqref{eq:bound4} hold.
  \end{lemma}
  The proof involves tedious, but straightforward, calculations, which
  we provide in Appendix~\ref{apdx:lemma:epsilon}. Concretely, we let
  $l=\frac{2-\sqrt{2}}{12}$, $r=\frac{7 - \sqrt{2}}{6}$, and
  $m=(l+r)/2$. Note that the (irrational) numbers $l$, $r$, and $m$
  are used only to obtain $\bfx$ from the given $\eps$-Nash equilibrium,
  and all rewards used to define the game $\Gamma_I$ are all fixed
  rational constants.
\end{proof}

\section{Nash Equilibrium in 4-player games}

In this section we prove that computing a stationary Nash equilibrium
in 4-player games is \SqrtSum-hard. To obtain this, we first construct
3-player games $G(a)$, parametrized by an integer $a \geq 1$ that has
a unique stationary Nash equilibrium with probabilities belonging to
$\QQ(\sqrt{a})$. 

\begin{definition}[The $\frac{1}{2}$-discounted 3-player game $G(a)$]
  For a given integer $a \geq 1$, the game $G(a)$ contains three nodes
  $s_1, s_2, s_3$ where $s_j$ is controlled by player~$j$ for
  $j \in \{1,2,3\}$, respectively. In state~$s_j$ player~$j$ is given
  the set of actions~$\{0,1\}$.  If player~$j$ chooses action~0, the
  game moves to state $s_{(j \bmod 3) + 1}$ and all players receive
  reward~$0$. If instead player~$j$ chooses action~1, each player
  receives a reward depending on~$j$, after which the game enters an
  absorbing state (i.e.\ a state where play never leaves) in which all
  players receive rewards~0.

  Player~$j$ obtains reward~$1$, player~$(j+1 \bmod 3) + 1$ obtains
  reward $2H$, and player~$(j \bmod 3) + 1$ obtains reward $4L$, where
  $L = 22 - \frac{162}{7}a, H = \frac{162}{7}a - 13$ and
  $\gamma = \frac{1}{2}.$ Clearly, $L < 1 < H$. The game is illustrated in
  Figure~\ref{fig:sqrt}~(i).
\label{def:G(a)}
\end{definition}

\begin{figure}[ht]
  \renewcommand\thesubfigure{\roman{subfigure}} \centering
    \begin{subfigure}[t]{0.3\textwidth}
    \centering
    \begin{tikzpicture}[scale=0.69]
    
    \draw[color=gray, very thick] (12,1) circle (0.6cm) node {\( s_1 \)};
    \draw[color=gray, very thick] (11,-0.5) circle (0.6cm) node {\( s_3 \)};
    \draw[color=gray, very thick] (13,-0.5) circle (0.6cm) node {\( s_2 \)};
    \draw[->, thick, Orange] (11.3,0) -- (11.7,0.46) ;
    \draw[->, thick, Orange] (12.3,0.46) -- (12.7,0);
    \draw[->, thick, Orange] (12.4,-0.5) -- (11.6,-0.5);
    \draw[->, thick, RoyalBlue] (12,1.6) -- (12,2.2) node[above]{\small $(1,4L,2H)$};
    \draw[->, thick, RoyalBlue] (10.7,-1) -- (10.4,-1.5) node[below]{\small $(4L,2H,1)$};
    \draw[->, thick, RoyalBlue] (13.3,-1) -- (13.6,-1.5) node[below]{\small $(2H,1,4L)$};

    \end{tikzpicture}    
    \caption{3-player Game $G(a)$}
  \end{subfigure}
  \hfill
  \begin{subfigure}[t]{0.3\textwidth}
    \centering
    \begin{tikzpicture}[scale=0.69]
    
    \draw[color=gray, very thick] (12,1) circle (0.6cm) node {\( s_1^i \)};
    \draw[color=gray, very thick] (11,-0.5) circle (0.6cm) node {\( s_3^i \)};
    \draw[color=gray, very thick] (13,-0.5) circle (0.6cm) node {\( s_2^i \)};
    \draw[->, thick, Orange] (11.3,0) -- (11.7,0.46) ;
    \draw[->, thick, Orange] (12.3,0.46) -- (12.7,0);
    \draw[->, thick, Orange] (12.4,-0.5) -- (11.6,-0.5);
    \draw[->, thick, RoyalBlue] (12,1.6) -- (12,2.2) node[above]{\small $(1,4L_i,2H_i, 1)$};
    \draw[->, thick, RoyalBlue] (10.7,-1) -- (10.4,-1.5) node[below]{\small $(4L_i,2H_i,1,4)$};
    \draw[->, thick, RoyalBlue] (13.3,-1) -- (13.6,-1.5) node[below]{\small $(2H_i,1,4L_i,2)$};

    \end{tikzpicture}    
    \caption{4-player Game $G_i$}
  \end{subfigure}
  \hfill
  \begin{subfigure}[t]{0.3\textwidth}
    \centering
    \begin{tikzpicture}[scale=0.69]
        \draw[color=gray, very thick] (5,0) circle (0.6cm) node {\( s_4 \)};
        \draw[->, thick, RoyalBlue] (4.4,0) -- (3.2,0) node[below] {\(r_0\)};
        \draw[thick, Orange] (5.6,0) --(6,0);
        \draw[thick, Orange] (6,0) --(6.1,0.2) --(6.2,0) --(6.1,-0.2) --cycle;
        \draw[->, thick, Orange] (6.2,0) --(7, 2) node[midway, above] {\footnotesize \(c_1\)};
        \draw[->, thick, Orange] (6.2,0) --(7, 0) node[midway, above] {\footnotesize \(c_i\)};
        
        \draw[->, thick, Orange] (6.2,0) --(7, -2) node[midway, above] {\footnotesize \(c_n\)};
        \draw[color=gray, very thick] (7.6, -2) circle (0.6cm) node{\(s_1^n\)};
        \fill (7.6,-0.8) circle (1pt);
        \fill (7.6,-1) circle (1pt);
        \fill (7.6,-1.2) circle (1pt);
        \draw[color=gray, very thick] (7.6, 0) circle (0.6cm) node{\(s_1^i\)};
        \fill (7.6,0.8) circle (1pt);
        \fill (7.6,1) circle (1pt);
        \fill (7.6,1.2) circle (1pt);        
        \draw[color=gray, very thick] (7.6, 2) circle (0.6cm)  node{\(s_1^1\)};
    \end{tikzpicture}    
    \caption{Game $G_0$}
  \end{subfigure}
  \caption{A 4-player Game.}
  \label{fig:sqrt}
\end{figure}
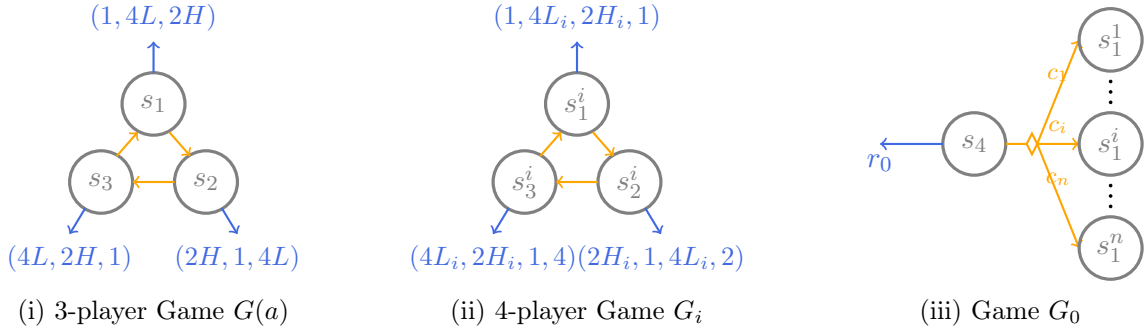

We may describe a stationary strategy profile in $G(a)$ by
$(x_1, x_2, x_3)$, where $x_j$ is the probability that player~$j$
chooses action~0 in state $s_j$. We next analyze the stationary Nash
equilibria in $G(a)$.
\begin{proposition}
  $G(a)$ has a unique stationary Nash equilibrium $(x_1,x_2,x_3)$ given by
  \[
    x_1 = x_2 = x_3 = \frac{35 - \frac{324}{7}a + 9\sqrt{a}}{2(22-
      \frac{162}{7}a - \frac{1}{8})} \enspace .
  \]
\end{proposition}

\begin{proof}
  First we show no player can use a pure strategy in a Nash
  equilibrium in $G(a)$. Due to the symmetry of $G(a)$, we only show
  this for player~1, and the argument for the other players is
  completely analogous. Consider a stationary Nash equilibrium
  $(x_1,x_2,x_3)$. Suppose that $x_1=0$. Then player~3 receives payoff
  $\frac{1}{2} \cdot \frac{1}{2} \cdot 2H=H/2$ by choosing action~$0$
  and payoff $1/2$ by choosing action~$1$. Since $H>1$, player~3 must
  choose action~0 with probability~$1$, i.e., $x_3=1$. Then, by the
  same argument, we get $x_2=1$, and then again that we must have
  $x_1=1$, contradicting the assumption. Suppose now that
  $x_1=1$. Then player~3 receives payoff $1/2$ by choosing action~$1$
  with probability~1, and strictly less than $1/2$ otherwise, since
  $L<1$. This means that player~3 must choose action~1 with
  probability~1, i.e.\ $x_3=0$. As in the case above, player~2 must
  then choose action~$0$ with probability~$1$, i.e.\ $x_2=1$. But then
  we must have $x_1=0$, contradicting the assumption.

  According to the above analysis, we only need to consider stationary
  strategy profiles where each player is playing a mixed strategy,
  i.e.\ $0<x_i<1$, for $i \in \{1,2,3\}$.  We next want to show
  $x_1 = x_2 = x_3$ in a Nash equilibrium. Since the strategies are
  all mixed, by the definition of a Nash equilibrium, both actions
  must give the same payoff to the controlling player. This means that
  the following three equations must hold:
  \begin{align*}
  \nonumber
    1 - \frac{1}{8} x_2 x_3 = (1 - x_2)H + x_2 (1 - x_3)L \\
    1 - \frac{1}{8} x_3 x_1 = (1 - x_3)H + x_3(1 - x_1)L\\
    1 - \frac{1}{8} x_1 x_2 = (1 - x_1)H + x_1(1 - x_2)L
  \end{align*}

  Define the function
  $f(x) = \frac{H - L}{\frac{1}{8} - L} + \frac{1 - H}{(\frac{1}{8} -
    L)x}$, and let us abbreviate this as $f(x) = k + \frac{b}{x}$.
  The equations above are then equivalent to:
  \[
    x_1 = f(x_3), x_2 = f(x_1), x_3 = f(x_2) \enspace .
  \]

  This means that $x_j = f(f(f(x_j)))$ must hold for
  $j \in \{1,2,3\}$.  In other words, each $x_j$ must be a root of the
  equation
  $x = f^3(x) = k + \frac{b}{k + \frac{b}{k + \frac{b}{x}}} = k +
  \frac{b}{k + \frac{bx}{kx + b}} = k + \frac{bkx + b^2}{k^2x + bx +
    kb}$.  However, there are at most two different roots of this
  equation, which means that at least two of the values $x_j$ must be
  equal. Without loss of generality, let us assume $x_1 = x_2$. Then
  we have $x_1 = x_2 = f(x_1) = f(x_2) = x_3$, and we can conclude
  $x_1 = x_2 = x_3$.

  We show in Appendix \ref{proof-uniqueness} that the equation
  $x = f^3(x)$ has a unique solution $x$ in the open interval $(0,1)$
  which is given as
  \begin{equation}
    x= \frac{35 - \frac{324}{7}a +
      9\sqrt{a}}{2(22-\frac{162}{7}a - \frac{1}{8})} \enspace .
  \label{eq:unique-solution}
  \end{equation}
\end{proof}

The problem \SqrtSum\/ is defined as follows: We are given positive
integers $a_1, \dots a_n$ and $t$, and are to decide whether
$\sum_{i=1}^n\sqrt{a_i} \leq t$. We next reduce \SqrtSum\/ to the
problem of computing a Nash equilibrium in a 4-player game.
\begin{theorem}
\label{thm:sqrt}
Computing a Nash equilibrium for 4-player discounted perfect
information stochastic games is \SqrtSum-hard.
\end{theorem}
\begin{proof}
  First we shall in polynomial time check whether
  $\sum_{i=1}^n\sqrt{a_i} = t$ and if so obtain the answer
  directly. The fact that this is possible has been attributed to
  Borodin~et~al.~\cite{JSC:BorodinFHT1985-decreasing-nesting-depth}
  by, for instance, Tiwari~\cite{JC:Tiwari92} and by Etessami and
  Yannakakis~\cite{SICOMP:EtessamiY2010-FIXP}. For completeness, we
  outline the argument below.

  In general, we may in polynomial time check whether an expression of
  the form $c_1\sqrt{r_1}+\dots+c_n\sqrt{r_n}$ is equal to~0, where
  $c_1,\dots,c_n$ are integers and $r_1,\dots,r_n$ are positive
  integers. If the radicals $\sqrt{r_1},\dots,\sqrt{r_n}$ are linearly
  independent over \QQ, the expression is equal to~0 only when
  $c_1=\dots=c_n=0$. The radicals are linearly independent if and only
  if for any pair $i\neq j$, the product $r_ir_j$ is not a perfect
  square. If in fact there exist $i\neq j$ such that $r_i r_j$ is a
  perfect square, which may be checked and found in polynomial time,
  we can rewrite the expression using the identity
  $c_i \sqrt{r_i} + c_j \sqrt{r_j} = (c_i + (\sqrt{r_ir_j}c_j)/r_i)\sqrt{r_i}$
  into an expression with fewer terms, and repeat.
  
  In the following we shall thus assume that
  $\sum_{i=1}^n\sqrt{a_i} = t$.  For each $a_i$, we construct the game
  $G_i$ based on $G(a_i)$ by setting rewards for player~4 as
  follows. If player~1 chooses action~1, player~4 obtains reward $1$.
  If player~2 chooses action~1, player~4 obtains reward $2$, and
  finally, if player~3 chooses action~1, player~4 obtains reward
  $4$. The game is illustrated in Figure~\ref{fig:sqrt}~(ii). Let
  $\overline{G_i}$ be the game obtained from $G_i$ by negating the
  rewards of player~4.

  The game $G$ is formed from picking, for every $i$, the game $G_i$
  or $\overline{G_i}$, together with $G_0$ as illustrated in
  Figure~\ref{fig:sqrt}~(iii). Player~4 only controls one state $s_4$
  and has two actions. If player~4 chooses action~0, the game moves
  with probability $c_i$ to the state $s_1^i$ in $G_i$.  Otherwise,
  player~4 obtains some reward $r_0$, and the game enters into an
  absorbing state with reward~0.

  We first analyze the payoff of player~4 when starting play in state $s_1^i$ of $G_i$.
  \begin{lemma}
    \label{lemma:sqrt}
    If the game $G_i$ starts at state $s_1^i$ and players~1,2, and 3
    follow their Nash equilibrium strategies, player~4 would obtain
    payoff $\frac{1}{2}(p_i + q_i \sqrt{a_i})$, where $p_i,q_i$ are
    rational numbers and $q_i > 0$.
  \end{lemma}
  The proof is given in Appendix~\ref{apdx:lemma:sqrt}. Based on
  Lemma~\ref{lemma:sqrt}, we define $C = \prod_{i=1}^nq_i$ and
  $D = \sum_{i=1}^nd_i$ where $d_i = C/q_i$. Let $c_i = d_i/D$, then
  $\sum_{i=1}^n c_i = 1$ and $0 < c_i < 1$. Finally, let $r_0$ be
  $\frac{1}{2} \sum_{i=1}^n c_ip_i + \frac{1}{2} (C/D)t$.

  Now, consider a stationary Nash equilibrium. If player~4 chooses
  action~1 starting in state $s_4$, the payoff obtained is exactly
  $V_1=\frac{1}{2}r_0 = \frac{1}{4} \sum_{i=1}^n c_ip_i + \frac{1}{4} (C/D)t$. If player~4 chooses action~0, the game moves with
  probability $c_i$ to the state $s_1^i$ in $G_i$. In this case,
  player~4 receives payoff
  \begin{equation}
    \begin{aligned}
    \nonumber
      V_0 = \frac{1}{2} \sum_{i=1}^n \frac{1}{2} \left(c_i(p_i + q_i\sqrt{a_i})\right) &= \frac{1}{4} \sum_{i=1}^n c_ip_i + \frac{1}{4} \sum_{i = 1}^n (d_i/D)q_i\sqrt{a_i} \\
        &= \frac{1}{4} \sum_{i=1}^nc_ip_i + \frac{1}{4} (C/D) \sum_{i=1}^n \sqrt{a_i}. \\
\end{aligned}
  \end{equation}

  Thus, $\sum_{i=1}^n\sqrt{a_i} < t$ if and only if $V_0 < V_1$ and
  $\sum_{i=1}^n\sqrt{a_i} > t$ if and only if $V_0 > V_1$. Since we
  assume $\sum_{i=1}^n\sqrt{a_i} \neq t$, player~4 must choose either
  action~0 or action~1 with probability~1, and this completes the
  reduction.
\end{proof}

\section{Conclusion}
We proved that for two-player perfect-information stochastic games,
the problem of computing a stationary Nash equilibrium is in
$\PPAD$. This leads to the interesting question of whether one may
develop an algorithm for this task, for instance based on Lemke's
algorithm, that may work well in practice. To complement this, we gave
an improved and simplified proof of $\PPAD$-hardness of computing
stationary $\eps$-Nash equilibria. While we give hardness for a
concrete value of $\eps$, it is still quite small. Improving this
bound further is another interesting problem. Probably the main
problem left open in our work is the precise computational complexity
of computing stationary Nash equilibria in perfect information games
with $3$ or more players. We proved the problem to be $\SqrtSum$-hard
in 4-player games, and leave open the question of whether the problem
may be $\FIXP$-hard.

\section*{Acknowledgements}
We thank Vidya Muthukumar for several helpful discussions.

\bibliographystyle{abbrv}
\bibliography{references}

\appendix

\section{Direct construction of a PL pseudo-circuit for solving the Selection Feasibility Program}
\label{sec:direct-PL-pseudo-circuit}

Here we present a direct construction of a PL pseudo-circuit that solves the form of feasibility program used in the proof of Proposition~\ref{prop:selection-circuit}. 

\begin{proposition}
    There exists a PL pseudo-circuit computing the correspondence
    $G \colon [0,1]^4 \rightrightarrows [0,1]$ defined by
    \[
        G(x_1,x_2,y_1,y_2) = \begin{cases}
            y_1 & \text{if } x_1 > x_2\\
            y_2 & \text{if } x_1 < x_2\\
            [\min(y_1,y_2),\max(y_1,y_2)] & \text{if } x_1=x_2
        \end{cases}    
    \]
\end{proposition}
\begin{proof}
    The circuit has an auxiliary input $z$ in addition to the inputs $x_1,x_2,y_1,y_2$. First the circuit computes $\Delta_i$ for $i \in [4]$ by
    \begin{align*}
        \Delta_1 = \max(0,\min(x_1-x_2,y_1-y_2)\\
        \Delta_2 = \max(0,\min(x_1-x_2,y_2-y_1)\\
        \Delta_3 = \max(0,\min(x_2-x_1,y_2-y_1)\\
        \Delta_4 = \max(0,\min(x_2-x_1,y_1-y_2)\\
    \end{align*}
    Note that $\Delta_i \geq 0$ for all $i \in [4]$, and for at most one $i$ we have that $\Delta_i > 0$.     
    The output of the circuit is then simply the transformation of the auxiliary input $z$ given by
    \[
        \tilde{z}=\max\left( \min(y_1,y_2), \min\left( \max(y_1,y_2), z+\Delta_1-\Delta_2+\Delta_3-\Delta_4\right)\right)
    \]
    We now analyze for which values of $z$ we have $\tilde{z}=z$,
    i.e., for which cases the auxiliary input $z$ is a fixed
    point. Note that any such $z$ must be contained in the interval
    $[\min(y_1,y_2),\max(y_1,y_2)]$ by construction of the circuit.
    
    In case $x_1=x_2$ we have $\Delta_i=0$ for all $i\in [4]$ and thus
    any $z \in [\min(y_1,y_2),\max(y_1,y_2)]$ is a fixed point. In
    case $x_1>x_2$ the only fixed point should be $y_1$, and this
    follows by noting that $\sgn(\Delta_1-\Delta_2) =
    \sgn(y_1-y_2)$. Finally, in case $x_1<x_2$ the only fixed point
    should be $y_2$, and this follows by noting that
    $\sgn(\Delta_3-\Delta_4) = \sgn(y_2-y_1)$.

\end{proof}

\section{Detailed analysis for Theorem~\ref{thm:epsilon}}
\subsection{\NOT\ gates}
\label{apdx：NOT-reduction}
\paragraph{\bf \NOT\ gates.}
To ensure Condition~\ref{not:con1} holds, we need to rule out any
$\eps$-Nash equilibrium having $p_u \leq l$ and $p_v<r$. To do this,
we ensure that by switching strategy to play action~1, player~1 would
increase the payoff by more than~$\eps$. In other words we will ensure
that
\[
  V = p_vV_1 + (1-p_v)V_0 \leq p_v \frac{1}{4} r_{\auxnot2} + (1-p_v)\left(\frac{1}{4}p_u+\frac{1}{12}\right) < \frac{1}{4} r_{\auxnot2} - \eps = V_1 - \eps
\]
which is equivalent to
\[
  p_u + \frac{1}{3} + \frac{4\eps}{1-p_v} < r_{\auxnot2}
\]
which in turn is implied by having
\[
  l + \frac{1}{3} + \frac{4\eps}{1-r} \leq r_{\auxnot2}
\]

To ensure Condition~\ref{not:con2} holds, we need to rule out any
$\eps$-Nash equilibrium having $p_u \geq r$ and $p_v > l$. To do this,
we ensure that by switching strategy to play action~0, player~1 would
increase the payoff by more than~$\eps$. In other words we will ensure
that
\[
  V = p_vV_1 + (1-p_v)V_0 = p_v \frac{1}{4} r_{\auxnot2} + (1-p_v)V_0 < V_0 - \eps
\]
which is equivalent to
\[
  r_{\auxnot2} < 4 V_0 - \frac{4\eps}{p_v}
\]
Using the lower bound on $V_0$ this is implied by having
\[
  r_{\auxnot2} < p_u - \frac{4\eps}{p_v} 
\]
which in turn is implied by having
\[
  r_{\auxnot2} \leq r - \frac{4\eps}{l} 
\]
Combining these, we conclude that whenever
$l + \frac{1}{3} + \frac{4\eps}{1-r} < r - \frac{4\eps}{l}$ we may
choose a rational constant reward $r_{\auxnot2}$ such that the \NOT\ gate
conditions are satisfied.

\subsection{\OR\ Gates}
\label{apdx：OR-reduction}
\paragraph{\bf \OR\ gates.}
To ensure Condition~\ref{or:con1} holds, we need to rule out any
$\eps$-Nash equilibrium having $p_{u_0} + p_{u_1} \leq 2l$ and $p_v>l$. To do this,
we ensure that by switching strategy to play action~0, player~1 would
increase the payoff by more than~$\eps$. In other words we will ensure
that
\[
  V = p_vV_1 + (1-p_v)V_0 \leq p_v \left(\frac{1}{8}(p_{u_0} + p_{u_1}) + \frac{1}{12}\right) + 
  (1 - p_v)\frac{1}{4}r_{\auxor2} < \frac{1}{4}r_{\auxor2} - \epsilon = V_0 - \eps
\]
which is equivalent to
\[
  \frac{1}{2}(p_{u_0} + p_{u_1}) + \frac{1}{3} + \frac{4\eps}{p_v} < r_{\auxor2}
\]
which in turn is implied by having
\[
  l + \frac{1}{3} + \frac{4\eps}{l} \leq r_{\auxor2}
\]

To ensure Condition~\ref{or:con2} holds, we need to rule out any
$\eps$-Nash equilibrium having $p_{u_0} + p_{u_1} \geq r$ and $p_v < r$. To do this,
we ensure that by switching strategy to play action~1, player~1 would
increase the payoff by more than~$\eps$. In other words we will ensure
that
\[
  V = p_vV_1 + (1-p_v)V_0 = p_v V_1 + (1-p_v)\frac{1}{4}r_{\auxor2} < V_1 - \eps
\]
which is equivalent to
\[
  r_{\auxor2} < 4 V_1 - \frac{4\eps}{1 - p_v}
\]
Using the lower bound on $V_1$ this is implied by having
\[
  r_{\auxor2} < \frac{1}{2}(p_{u_0} + p_{u_1}) - \frac{4\eps}{1 - p_v} 
\]
which in turn is implied by having
\[
  r_{\auxor2} \leq \frac{1}{2}r - \frac{4\eps}{1 - r} 
\]
Combining these, we conclude that whenever
$l + \frac{1}{3} + \frac{4\eps}{l} < \frac{1}{2}r - \frac{4\eps}{1-r}$ we may
choose a rational constant reward $r_{\auxor2}$ such that the \OR\ gate
conditions are satisfied.

\subsection{\PURIFY\ Gates}
\label{apdx：PURIFY-reduction}
\paragraph{\bf \PURIFY\ gates.}
To ensure Condition~\ref{pur:con1} holds, we need to rule out any $\eps$-Nash equilibrium having that
$p_u \leq m$ and $p_{v_0} > l$. To do this,
we ensure that by switching strategy to play action~0, player~1 would
increase the payoff by more than~$\eps$. In other words we will ensure
that
\[
  V = p_{v_0}V_1 + (1-p_{v_0})V_0 \leq p_{v_0} \left(\frac{1}{4}p_u + \frac{1}{12}\right) + 
  (1 - p_{v_0})\frac{1}{4}r_{\auxpur20} < \frac{1}{4}r_{\auxpur20} - \epsilon = V_0 - \eps
\]
which is equivalent to
\[
  p_u + \frac{1}{3} + \frac{4\eps}{p_{v_0}} < r_{\auxpur20}
\]
which in turn is implied by having
\[
  m + \frac{1}{3} + \frac{4\eps}{l} \leq r_{\auxpur20}
\]

To ensure Condition~\ref{pur:con2} holds, we need to rule out any $\eps$-Nash equilibrium having that  $p_u \geq r$ and $p_{v_0} < r$. To do this,
we ensure that by switching strategy to play action~1, player~1 would
increase the payoff by more than~$\eps$. In other words we will ensure
that
\[
  V = p_{v_0}V_1 + (1-p_{v_0})V_0 = p_{v_0} V_1 + (1-p_{v_0})\frac{1}{4}r_{\auxpur20} < V_1 - \eps
\]
which is equivalent to
\[
  r_{\auxpur20} < 4 V_1 - \frac{4\eps}{1 - p_{v_0}}
\]
Using the lower bound on $V_1$ this is implied by having
\[
  r_{\auxpur20} < p_u - \frac{4\eps}{1 - p_{v_0}} 
\]
which in turn is implied by having
\[
  r_{\auxpur20} \leq r - \frac{4\eps}{1 - r} 
\]
Combining these, we conclude that whenever
$m + \frac{1}{3} + \frac{4\eps}{l} < r - \frac{4\eps}{1-r}$ we may
choose a rational constant reward $r_{\auxpur20}$.

To ensure Condition~\ref{pur:con3} holds, we need to rule out any $\eps$-Nash equilibrium having that
$p_u \leq l$ and $p_{v_1} > l$. To do this,
we ensure that by switching strategy to play action~0, player~1 would
increase the payoff by more than~$\eps$. In other words we will ensure
that
\[
  V = p_{v_1}V_1 + (1-p_{v_1})V_0 \leq p_{v_1} \left(\frac{1}{4}p_u + \frac{1}{12}\right) + 
  (1 - p_{v_1})\frac{1}{4}r_{\auxpur21} < \frac{1}{4}r_{\auxpur21} - \epsilon = V_0 - \eps
\]
which is equivalent to
\[
  p_u + \frac{1}{3} + \frac{4\eps}{p_{v_1}} < r_{\auxpur21}
\]
which in turn is implied by having
\[
  l + \frac{1}{3} + \frac{4\eps}{l} \leq r_{\auxpur21}
\]

To ensure Condition~\ref{pur:con4} holds, we need to rule out any $\eps$-Nash equilibrium having that  $p_u \geq m$ and $p_{v_1} < r$. To do this,
we ensure that by switching strategy to play action~1, player~1 would
increase the payoff by more than~$\eps$. In other words we will ensure
that
\[
  V = p_{v_1}V_1 + (1-p_{v_1})V_0 = p_{v_1} V_1 + (1-p_{v_1})\frac{1}{4}r_{\auxpur21} < V_1 - \eps
\]
which is equivalent to
\[
  r_{\auxpur21} < 4 V_1 - \frac{4\eps}{1 - p_{v_1}}
\]
Using the lower bound on $V_1$ this is implied by having
\[
  r_{\auxpur21} < p_u - \frac{4\eps}{1 - p_{v_1}} 
\]
which in turn is implied by having
\[
  r_{\auxpur21} \leq m - \frac{4\eps}{1 - r} 
\]
Combining these, we conclude that whenever
$l + \frac{1}{3} + \frac{4\eps}{l} < m - \frac{4\eps}{1-r}$ we may
choose a rational constant reward $r_{\auxpur21}$ such that the \PURIFY\ gate
conditions are satisfied.

\subsection{Proof of Lemma~\ref{lemma:epsilon}}
\label{apdx:lemma:epsilon}
\begin{proof}
    Following from all above conditions     \eqref{eq:bound1}, \eqref{eq:bound2},
    \eqref{eq:bound3}, \eqref{eq:bound4} and
    $0 \leq l < m < r \leq 1$, we
    have
    \begin{equation}
    \begin{aligned}
    \left(\frac{1}{1-r} + \frac{1}{l}\right)2\epsilon &<
      \min
      \left\{
      \frac{1}{2}r - \frac{1}{2}m - \frac{1}{6},
      \frac{1}{2}m - \frac{1}{2}l - \frac{1}{6},
      \frac{1}{4}r - \frac{1}{2}l - \frac{1}{6},
      \frac{1}{2}r - \frac{1}{2}l - \frac{1}{6}
                                            \right\}\\
      &=\min
      \left\{
      \frac{1}{2}r - \frac{1}{2}m - \frac{1}{6},
      \frac{1}{2}m - \frac{1}{2}l - \frac{1}{6},
      \frac{1}{4}r - \frac{1}{2}l - \frac{1}{6}\right\}                                      
    \end{aligned}
    \end{equation}

    We are looking for an upper bound of $\epsilon,$ thus,
    \begin{equation}
    \begin{aligned}
    2\epsilon &<
      \max_{l,r,m} \left[\frac{l(1-r)}{l+1-r}\min
      \left\{
      \frac{1}{2}r - \frac{1}{2}m - \frac{1}{6},
      \frac{1}{2}m - \frac{1}{2}l - \frac{1}{6},
      \frac{1}{4}r - \frac{1}{2}l - \frac{1}{6} \right\} \right]\\
      &\overset{1}{=}\max_{l,r,m} \left[\frac{l(1-r)}{l+1-r}\min
      \left\{
      \frac{1}{4}r - \frac{1}{4}l - \frac{1}{6},
      \frac{1}{4}r - \frac{1}{2}l - \frac{1}{6} \right\} \right]\\
             &=\max_{l,r} \frac{l(1-r)}{l+1-r}\left(\frac{1}{4}r - \frac{1}{2}l - \frac{1}{6}\right)\\
             & \overset{2}{=} \frac{3 - 2 \sqrt{2}}{144}.                             
    \end{aligned}
  \end{equation}

Step~1 holds because the maximum of minimum between $\frac{1}{2}r - \frac{1}{2}m - \frac{1}{6}$ and $\frac{1}{2}m - \frac{1}{2}l - \frac{1}{6}$ is attained if and only if the two values are equal. The following shows the detailed steps for Step~2. Our goal is to find the maximum of the function $f(l,r)$ where $0 \leq l < r \leq 1.$ The function $f(l,r)$ is simplified as follows.

  \[
  \begin{aligned}
    f(l,r) &= \frac{1}{12} \frac{(3r - 6l - 2)l(1-r)}{1-r+l} \\
    &\overset{ \scriptstyle t=1-r+l}{=} \frac{1}{12} \frac{(-3t-3l+1)l(t-l)}{t} = \frac{1}{12}\left(\frac{3l^3-l^2}{t} - 3lt + l\right)
  \end{aligned}
  \]

  From the ranges of $l$ and $r$, it follows that $0 \leq l \leq t < 1.$ We treat $l$ in the above function as a constant and discuss the maximum value of the function under different cases based on the value of $l.$ Thus, we rewrite the function $f$ as $h(t),$ with $h'(t)$ representing the derivative of $h$ with respect to the variable $t.$
  \begin{enumerate}
  \item $0 \leq l \leq \frac{1}{3}:$ In this case, $h'(t) = \frac{1}{12}(-\frac{3l^3 - l^2}{t^2} - 3l).$ The function $h'(t) \geq 0$ if and only if $3t^2 \leq l - 3l^2.$ We conduct a more detailed discussion on the values of $l.$
    \begin{enumerate}[label=\roman*.]
    \item $0 \leq l \leq \frac{1}{6}:$ When $t = \sqrt{\frac{l}{3} - l^2},$ we obtain the maximum value of $h(t)$ which equals to $\frac{l}{12} - \frac{l}{2}\sqrt{\frac{l}{3}-l^2}.$
    \item $\frac{1}{6} < l \leq \frac{1}{3}:$ It's easy to see that $3t^2 \geq 3l^2 > l - 3l^2.$ This immediately implies that the maximum value of $h$ can be achieved at $t = l$ where the maximum equals $0.$
    \end{enumerate}
    \item $\frac{1}{3} < l < 1:$ In this case, $h'(t) < 0.$ The maximum value of $h$ can be achieved at $t = l.$ At this point, the maximum value is $0.$
  \end{enumerate}

  We now proceed to determine the maximum of $\frac{l}{12} - \frac{l}{2}\sqrt{\frac{l}{3}-l^2}$. Using $z = 1 - 6l,$ we can simplify the expression to $\frac{(1-z)(1-\sqrt{1-z^2})}{72}.$ It's not difficult to compute that the maximal value is $\frac{3 - 2 \sqrt{2}}{144},$ which is attained at $1-z = 1-\sqrt{1-z^2}.$ As a result, $l$ takes the value $\frac{2-\sqrt{2}}{12}$ while $r$ takes the value $\frac{7 - \sqrt{2}}{6}.$ 
\end{proof}

\section{Details in Theorem~\ref{thm:sqrt}}
\subsection{Proof of uniqueness of root}
  \label{proof-uniqueness}
  \begin{proof}
  We know $x$ is one of the roots of equation:
  \begin{equation}
    \label{eq:prob}
    \begin{aligned}
      1 - \gamma^3 x^2 = (1 - x) H + x(1 - x) L \\
      \iff (L - \gamma^3)x^2 + (H - L)x + (1 - H) = 0
    \end{aligned}
  \end{equation}

  Let $\Delta$ be $(H-L)^2 - 4(L - \gamma^3)(1 - H)$. Then $x$ must be equal to $x_1$ or $x_2$.
  \begin{equation}
    \begin{aligned}
      x_1 = \frac{L - H - \sqrt{\Delta}}{2(L - \gamma^3)}, x_2 = \frac{L - H + \sqrt{\Delta}}{2(L - \gamma^3)}
    \end{aligned}
  \end{equation}

  Reviewing our settings in the game -- that is, $L = 22 - \frac{162}{7}a, H = \frac{162}{7}a - 13, \gamma = \frac{1}{2}$. Then $\Delta = (H + L)^2 - 4L + \frac{1 - H}{2} = 77 - \frac{7}{2}L = 81a$.

  Thus, we have
  \begin{equation}
    \begin{aligned}
      x_1 = \frac{35 - \frac{324}{7}a - 9\sqrt{a}}{2(22 - \frac{162}{7}a - \frac{1}{8})}, x_2 = \frac{35 - \frac{324}{7}a + 9\sqrt{a}}{2(22 - \frac{162}{7}a - \frac{1}{8})}
    \end{aligned}
  \end{equation}

  Define $g(x) = (L - \gamma^3)x^2 + (H - L)x + (1 - H)$, using $a \geq 1$ we can show
  \begin{equation}
    \begin{aligned}
      g(0) = 1 - H = 14 - \frac{162}{7} a < 0 \\
      g(1) = 1 - \gamma^3 > 0
    \end{aligned}
  \end{equation}

  Additionally, $L - \gamma^3 < L \leq 22 - \frac{162}{7} < 0$, then $g(\infty) < 0$. In other words, the smaller root is in $(0,1)$ and another root is larger than $1$. Thus, $x$ is exactly the smaller root of $x_1$ and $x_2$.

  Finally, we can conclude $x = \frac{35 - \frac{324}{7}a + 9\sqrt{a}}{2(22 - \frac{162}{7}a - \frac{1}{8})}$.
  \end{proof}

\subsection{Proof of Lemma~\ref{lemma:sqrt}}
\label{apdx:lemma:sqrt}
  \begin{proof}
    For the sake of simplicity, we use $x$ instead of
    $x^i$ in the following analysis. Define $v$ as the
    valuation of rewards player~4 can obtain when the
    game $G_i$ starts at the state $s_1^i$ and
    players~1,2,3 follow the equilibrium strategy. By the
    construction of $G_i$, we have
    \begin{equation}
    v = (1 - x) + x(1 - x) + x^2(1 - x) + \frac{1}{8} x^3 v
    \end{equation}

    Thus $v$,
    the reward of player~4, equals to
    $\frac{8 - 8x^3}{8 - x^3}$. We analyze $v$ based on
    the value of $a_i.$ Selecting a candidate $a_i$, the
    final result $v$ may fall into one of three possible
    cases. If $a_i$ is a perfect square, $v$ would be a
    rational number. Then, naturally, it can be
    rewritten in another form
    $\frac{v}{\sqrt{a_i}}\sqrt{a_i}$ where
    $\frac{v}{\sqrt{a_i}}$ is a rational number.
    If $a_i$ is \emph{not} a square-free integer, we
    denote $a_i$ as $b^2d$ where $d$ is a square-free
    integer. Due to $x \in \mathbb{Q}(\sqrt{d}) \setminus
    \mathbb{Q}$, $v$ must be in $\mathbb{Q}(\sqrt{d})$
    because $v$ is obtained from $x$ through addition
    and multiplication. Moreover,
    $v \not \in \mathbb{Q}$. Suppose $v$ is a rational
    number, we can get $x^3 = \frac{8v - 8}{v - 8}$ from
    $v = \frac{8 - 8x^3}{8 - x^ 3}$, which means $x^3$ is
    a rational number. Denote $x$ as $p + q \sqrt{d}$, we
    have $x^3 = (dq^3 + 3 p^2 q)\sqrt{d} + 3 d p q^2 +
    p^3$. Because of $x \in \mathbb{Q}(\sqrt{d})
    \setminus \mathbb{Q}$, which implies $q \neq 0$ and
    $dq^3 + 3p^2q \neq 0$, $x^3$ cannot be a rational
    number. Based on the above analysis, we denote $v$
    as $p + q\sqrt{d}$ where $q \neq 0$. Similarly, we
    can rewrite $v$ as $p + \frac{q}{b}\sqrt{b^2d}$. If
    $a_i$ is a square-free integer, we can directly use
    the previous conclusion by simply setting $b=1$. At
    this point, $v$ is of form $p + q\sqrt{a_i}$.

    In summary, regardless of which of the three cases
    it is, we can calculate player~4's payoff and
    express it in the form $p _i+ q_i\sqrt{a_i}$ where
    $q_i \neq 0$. If $q_i < 0$, we can simply set
    player~4's reward to its negation in $G_i$, thereby
    ensuring that $q_i > 0$. Notice that $p_i$ and $q_i$
    can be efficiently computed by directly using
    equation~\eqref{eq:unique-solution}.
  \end{proof}
\end{document}